\documentclass[11pt]{amsart}
\usepackage{amsmath}
\usepackage{amsfonts}
\usepackage{amssymb}
\usepackage{enumerate}
\usepackage{graphicx}
\usepackage{setspace}

\newtheorem{thm}{Theorem}[section]

\newtheorem{cor}[thm]{Corollary}

\newtheorem{remark}[thm]{Remark}
\newtheorem{lemma}[thm]{Lemma}
\newtheorem{prop}[thm]{Proposition}

\newtheorem{exam}[thm]{Example}
\newtheorem{defn}[thm]{Definition}

\newcommand{\bra}[1]{\langle #1 |}
\newcommand{\ket}[1]{| #1 \rangle}
\newcommand{\braket}[2]{\langle #1 | #2 \rangle}
\newcommand{\ketbra}[2]{| #1 \rangle\langle #2 |}
\newcommand{\Tr}{{\rm Tr}}
\newcommand{\bb}[1]{\mathbb{#1}}
\newcommand{\cl}[1]{\mathcal{#1}}

\begin{document}

\title[A Family Of Norms With Applications In Quantum Information Theory]{A Family Of Norms With Applications In \\ Quantum Information Theory}

\author[N.~Johnston, D.~W.~Kribs]{Nathaniel Johnston$^1$ and David~W.~Kribs$^{1,2}$}
\address{$^1$Department of Mathematics \& Statistics, University of Guelph,
Guelph, ON, Canada N1G 2W1}
\address{$^2$Institute for Quantum Computing, University of Waterloo, Waterloo, ON, Canada
N2L 3G1}

\begin{abstract}
We consider a family of vector and operator norms defined by the Schmidt decomposition theorem for quantum states. We use these norms to tackle two fundamental problems in quantum information theory: the classification problem for $k$-positive linear maps and entanglement witnesses, and the existence problem for non-positive partial transpose bound entangled states. We begin with an analysis of the norms, showing that the vector norms can be explicitly calculated, and we derive several inequalities in order to bound the operator norms and compute them in special cases. We then use the norms to establish what appears to be the most general spectral test for $k$-positivity currently available, showing how it implies several other known tests as well as some new ones. Building on this work, we frame the NPPT bound entangled problem as a concrete problem on a specific limit, specifically that a particular entangled Werner state is bound entangled if and only if a certain norm inequality holds on a given family of projections.
\end{abstract}

\maketitle

\section{Introduction}

Entanglement theory lies at the heart of investigations in quantum information. One of the most basic tools in this theory is the Schmidt decomposition theorem for quantum states \cite{NC00}.  In this paper we consider a family of vector and operator norms defined by the Schmidt theorem. We conduct the first in-depth analysis of these norms and we use them to tackle two central problems in quantum information theory: the classification problem for $k$-positive linear maps and entanglement witnesses, and the existence problem for non-positive partial transpose (NPPT) bound entangled states.


The family of norms generalize the standard Euclidean and operator norms and can be regarded as the local analogues of these norms. The vector norms have recently appeared in \cite{CKo09,CKS09} as a tool for testing $k$-positivity of linear maps. The operator and vector $2$-norms have appeared in literature related to NPPT bound entangled states \cite{PPHH07,DSSTT00}. We begin the paper with a systematic study of the norms. After deriving their basic properties, we focus on bounding them in a variety of ways and computing them in special cases. We also show how these norms are related to quantum fidelity and trace distance, and regularized relative entropy of entanglement \cite{VP98,VPRN97}.



Once we have developed a number of tools to handle these norms, we establish what appears to be the most general spectral test for $k$-positivity, and hence $k$-entanglement witnesses, currently available. We derive this test as an abstract machine, and then apply it to concrete situations. We use it to reproduce the recently-developed tests of Chru\'{s}ci\'{n}ski and Kossakowski \cite{CKo09}, which in turn imply the tests of Takesaki and Tomiyama \cite{TT83} and Benatti, Floreanini, and Piani \cite{BFP04}. We also show it implies the test of Kuah and Sudarshan \cite{KSu05}. And we find a number of new tests for $k$-positivity, including a complete characterization in the case of two distinct eigenvalues.

We then explore a connection between $k$-positivity of linear maps and the existence of NPPT bound entangled states. Separable states are bound entangled, as are states with positive partial transpose \cite{P96,H97,HHH98}. However, it is unknown whether or not there exist NPPT bound entangled states, and the existence of such states would exhibit a fundamentally new type of entanglement \cite{DSSTT00}. We apply our analysis of the operator norms on projections to the crucial case of Werner states \cite{W89}, and we frame the NPPT bound entangled problem as a concrete calculus problem.



The paper is arranged as follows. In Section~\ref{sec:prelim} we present our notation and terminology and introduce the reader to the required notions from operator theory and quantum information. In Section~\ref{sec:vectorNorms} we will define and explore vector norms, which can be thought of as measuring how close pure states are to having a given Schmidt rank. We will then define the operator norms in Section~\ref{sec:MatrixNorms}, which apply to arbitrary mixed states and recover the vector norms in the case of pure states. We will see that the operator norms are very difficult to calculate in general, so we will develop several inequalities to bound them in various situations.

Section~\ref{sec:SpectralInequalities} will focus on the problem of determining whether or not a given operator is $k$-block positive -- in the language of quantum information this is the problem of determining whether or not that operator is a $k$-entanglement witness. We show how the $k$th  operator norms can be used to derive several testable conditions for $k$-positivity, and we derive a complete characterization in the case when the operator has two distinct eigenvalues. In Section~\ref{sec:WernerStates} we will apply our $k$-block positivity tests to Werner states and show that a particular NPPT Werner state is bound entangled if and only if a certain limit involving the operator norms is satisfied.

\section{Preliminaries}\label{sec:prelim}

We will use $\mathcal{H}$ to denote a finite-dimensional complex Hilbert space and $\mathcal{L}(\mathcal{H})$ to denote the set of linear operators on $\mathcal{H}$. When the dimension of the Hilbert space is important, we will denote it $\cl{H}_n$, where $n$ is its dimension. Similarly, $id_n$ will represent the identity map on $\cl{L}(\cl{H}_n)$. Of particular interest in quantum information is the case when $\cl{H}$ is a \emph{bipartite system} -- a tensor product of two smaller Hilbert spaces $\cl{H} = \cl{H}_n \otimes \cl{H}_m$. We will assume for the sake of brevity throughout the paper that $m \leq n$. A vector $\ket{v} \in \cl{H}$ is denoted using Dirac bra-ket notation, with $\bra{v} := \ket{v}^*$. Whenever we use this bra-ket notation, it will be assumed that $\ket{v}$ is such that $\big\|\ket{v}\big\| = 1$ and so $\ket{v}$ represents a pure state (or more correctly the associated state is given by the rank one projection $\ketbra{v}{v}$). We will denote the computational basis vectors (i.e., the vectors with $1$ in the $i^{th}$ component and $0$ in all other components) by $\{\ket{e_i}\}$.

If $X \in \cl{L}(\cl{H})$ is positive then we will write $X \geq 0$ or $X \in \cl{L}(\cl{H})^{+}$. A (mixed) quantum state is represented by a \emph{density operator} $\rho \geq 0$ that satisfies $\Tr(\rho) = 1$. Whenever lowercase Greek letters like $\rho$ or $\sigma$ are used, it is assumed that they are density operators. General operators will be represented by uppercase letters like $X$ and $Y$.

Given a linear map $\Phi : \mathcal{L}(\cl{H}_n) \rightarrow \cl{L}(\cl{H}_m)$, we can define its dual map $\Phi^\dagger : \mathcal{L}(\mathcal{H}_m) \rightarrow \mathcal{L}(\mathcal{H}_n)$ via the Hilbert-Schmidt inner product $\Tr(\Phi(X)Y) = \Tr(X \Phi^\dagger(Y))$. The map $\Phi$ is said to be:
\begin{itemize}
    \item \emph{Hermicity-preserving} if $\Phi(X)^* = \Phi(X)$ whenever $X^* = X$.
    \item \emph{Positive} if $\Phi(X) \geq 0$ whenever $X \geq 0$.
    \item \emph{$k$-positive} if $(id_k \otimes \Phi)(X) \geq 0$ whenever $X \in (\cl{L}(\cl{H}_k) \otimes \cl{L}(\cl{H}_n))^+$.
    \item \emph{Completely positive} if $\Phi$ is $k$-positive for all $k \in \bb{N}$.
\end{itemize}

\noindent A theorem of Choi says that $n$-positivity of $\Phi$ is equivalent to complete positivity of $\Phi$ \cite{C75,Paulsentext}. Furthermore, $\Phi$ is completely positive if and only if $(id_n \otimes \Phi)(E) \geq 0$, where $E := \frac{1}{n}\sum_{i,j=1}^n\ketbra{e_i}{e_j} \otimes \ketbra{e_i}{e_j}$. The matrix form for the operator $(id_n \otimes \Phi)(E)$ is referred to as the \emph{Choi matrix} of $\Phi$. In fact, the Choi matrix defines an isomorphism (known as the \emph{Choi-Jamiolkowski isomorphism} \cite{J72}) between linear maps $\Phi : \mathcal{L}(\cl{H}_n) \rightarrow \cl{L}(\cl{H}_m)$ and operators $X \in \cl{L}(\cl{H}_n) \otimes \cl{L}(\cl{H}_m)$. Under this isomorphism, the Hermicity-preserving maps $\Phi$ correspond to the Hermitian operators $X$.  In keeping with the terminology of \cite{S08,SSZ09}, we will say that a Hermitian operator $X = X^* \in \cl{L}(\cl{H}_n) \otimes \cl{L}(\cl{H}_m)$ is \emph{$k$-block positive} if the associated linear map is $k$-positive.

Several connections will be made between these norms and other more well-known norms. In particular, for an operator $X \in \cl{L}(\cl{H}_n)$ it will be useful to be familiar with the \emph{Ky Fan $k$-norm} \cite{HJ91} of $X$, given by $\big\|X\big\|_k := \sum_{i=1}^k s_i$, where $s_1 \geq \cdots \geq s_n$ are the singular values of $X$. Note that the smallest of the Ky Fan norms, the Ky Fan $1$-norm, is equal to the operator norm. The largest of the Ky Fan norms, the Ky Fan $n$-norm, is equal to the \emph{trace norm} because it can be written as $\big\|X\big\|_{n} = \Tr(|X|)$, where $|X| := \sqrt{X^*X}$ is the absolute value of $X$.

Some related distance measures that are used frequently in quantum information are the {\em trace distance} $\delta$ and the \emph{quantum fidelity} $F$ between two density operators $\rho, \sigma \in \cl{L}(\cl{H})$:
\begin{align*}
  \delta(\rho,\sigma) & := \frac{1}{2}\Tr(|\rho - \sigma|), \\
  F(\rho,\sigma) & := \Big( \Tr\big( \sqrt{ \sqrt{\rho} \sigma \sqrt{\rho} } \big) \Big)^2.
\end{align*}

The trace distance can be thought of as the distance between $\rho$ and $\sigma$, and the quantum fidelity can be interpreted as the amount of overlap between them. They both simplify when their inputs are pure states, and in particular
\begin{align}\label{eq:traceDistPure}
  \delta(\ketbra{v}{v},\ketbra{w}{w}) & = \sqrt{1 - |\braket{v}{w}|^2}, \\ \label{eq:fidelityPure}
  F(\ketbra{v}{v},\sigma) & = \bra{v}\sigma\ket{v}.
\end{align}

\noindent These will be useful tools for providing interpretations of the  vector and operator norms that will be introduced.

\subsection{Schmidt Rank and Schmidt Number}\label{sec:schmidt}

The Schmidt Decomposition Theorem \cite[Section 2.5]{NC00} is a basic tool in quantum information theory. It states that if $\ket{v} \in \cl{H}_n \otimes \cl{H}_m$ then there exists $k \leq m$ (recall that $m \leq n$ by assumption) and orthonormal sets of vectors $\{ \ket{u_1}, \ket{u_2}, \ldots, \ket{u_k} \} \subset \cl{H}_n$ and $\{ \ket{v_1}, \ket{v_2}, \ldots, \ket{v_k} \} \subset \cl{H}_m$ such that
  \begin{align}\label{eq:schmidtdecomp}
    \ket{v} = \sum_{i=1}^k{\alpha_i \ket{u_i} \otimes \ket{v_i}}
  \end{align}

\noindent for some non-negative real constants $\{\alpha_i\}$.

The standard proof of the Schmidt Decomposition works by noticing that there is an isomorphism between $\cl{H}_n \otimes \cl{H}_m$ and $\cl{L}(\cl{H}_n,\cl{H}_m)$ given by associating a vector $\ket{u_i} \otimes \ket{v_i}$ with the operator $\ket{u_i}\overline{\bra{v_i}}$ and extending linearly. We will denote the operator associated with the vector $\ket{v}$ by $A_v$. Applying the singular value decomposition to $A_v$ gives the Schmidt Decomposition of $\ket{v}$.

In the Schmidt Decomposition~\eqref{eq:schmidtdecomp} of $\ket{v}$, the least number of terms required in the summation is known as the \emph{Schmidt rank} of $\ket{v}$, denoted $SR(\ket{v})$. It follows that the Schmidt rank of $\ket{v}$ is equal to the number of non-zero singular values of the operator to which $\ket{v}$ is associated (i.e., its rank). Similarly, the $\alpha_i$'s are exactly the singular values of $A_v$. Because the singular value decomposition is easy to compute, so are the Schmidt rank and the Schmidt Decomposition of an arbitrary pure state $\ket{v}$.

The following useful recent result of Cubitt, Montanaro and Winter \cite{CW08}  provides a tight bound on the dimension of subspaces consisting entirely of vectors with high Schmidt rank.
\begin{thm}\label{thm:CMW08}
  The maximum dimension of a subspace $\cl{S} \subseteq \cl{H}_n \otimes \cl{H}_m$ such that $SR(\ket{v}) \geq k$ for all $\ket{v} \in \cl{S}$ is given by $(n - k + 1)(m - k + 1)$.
\end{thm}
Not only is $(n - k + 1)(m - k + 1)$ shown to be an upper bound on the dimension of such subspaces, but an explicit method of construction is given that produces such a subspace that attains the bound.

In analogy with the Schmidt rank for pure states, the \emph{Schmidt number} \cite{TH00} of a mixed state $\rho$ is defined to be the least natural number $k$ such that $\rho$ can be written as
\[
    \rho = \sum_i p_i \ketbra{v_i}{v_i},
\]

\noindent where $SR(\ket{v_i}) \leq k$ for all $i$ and $\{ p_i \}$ forms a probability distribution. The Schmidt number of a state can be thought of as a rough measure of how entangled that state is. One case that is of particular interest is when $SN(\rho) = 1$, in which case $\rho$ is said to be \emph{separable}. It is not difficult to check that $\rho$ is separable if and only if it can be written as $\rho = \sum_i X_i \otimes Y_i$ for some $\big\{ X_i \big\}, \big\{ Y_i \big\} \geq 0$.

It has been shown \cite{HHH96} that if $m = 2$ and $n = 2$ or $n = 3$ then $\rho$ is separable if and only if $(id_n \otimes T)(\rho) \geq 0$, where  $T$ is the transpose map. The fact that the transpose map can be used to determine separability in small dimensions has led to the study of \emph{positive partial transpose (PPT)} states in arbitrary dimensions \cite{P96}, which are density operators $\rho$ such that $(id_n \otimes T)(\rho) \geq 0$. Throughout the rest of this paper, we will write the partial transpose operation $(id_n \otimes T)(\rho)$ as $\rho^\Gamma$.

Finally, we present without proof a well-known result that shows an intricate connection between $k$-block positivity of operators and the Schmidt number of operators. (For instance, it easily follows from the recently-explored dual cone relationship of $k$-positivity and Schmidt number \cite{S08,SSZ09}.)
\begin{lemma}\label{prop:kPos}
    Let $X \in \cl{L}(\cl{H}_n) \otimes \cl{L}(\cl{H}_m)$. Then $X$ is $k$-block positive if and only if
    \begin{align*}
        \Tr(X\rho) \geq 0 \quad \forall \, \rho \in \cl{L}(\cl{H}_n) \otimes \cl{L}(\cl{H}_m) \text{ with } SN(\rho) \leq k.
    \end{align*}
\end{lemma}

\section{Vector Norms}\label{sec:vectorNorms}

With the Schmidt Decomposition in hand, we can define a new family of norms that generalize the standard Euclidean norm.
\begin{defn}
  Let $\ket{v} \in \cl{H}_n \otimes \cl{H}_m$ and let $1 \leq k \leq m$. Then we define the {\em $k$th vector norm} of $\ket{v}$, denoted $\big\| \ket{v} \big\|_{s(k)}$, by
  \begin{align*}
    \big\| \ket{v} \big\|_{s(k)} & := \sup_{\ket{w}} \Big\{ \big| \braket{w}{v} \big| : SR(\ket{w}) \leq k \Big\}.
  \end{align*}
\end{defn}

\begin{remark}\label{Schmidtvectorrem1}
{\rm Note that even though this definition is only stated for unit vectors $\ket{v}$, it extends in the obvious way to a norm on all of $\cl{H}_n \otimes \cl{H}_m$. These norms were very recently considered independently in \cite{CKo09,CKS09} as a tool for detecting $k$-block positivity of operators. We shall return to this topic below. Intuitively, the vector $k$-norm has a simple interpretation in quantum information, as it can be viewed as a measure of how close a given state is to a state of Schmidt rank at most $k$. This is made more precise below. The case of $k = m$ is very familiar: $\big\| \ket{v} \big\|_{s(m)} = \big\| \ket{v} \big\|$. Also, it is clear from the definition that $\big\| \ket{v} \big\|_{s(k)} \leq \big\| \ket{v} \big\|$ for all $k$, and it is not difficult to see that we have an increasing family of norms leading up to the standard Euclidean norm:
\begin{align*}
  \big\| \ket{v} \big\|_{s(1)} \leq \big\| \ket{v} \big\|_{s(2)} \leq \cdots \leq \big\| \ket{v} \big\|_{s(m-1)} \leq \big\| \ket{v} \big\|.
\end{align*}}
\end{remark}

The first result shows that this norm is not particularly difficult to calculate.

\begin{thm}\label{thm:vectorNorm}
  Let $\ket{v} \in \cl{H}_n \otimes \cl{H}_m$ have Schmidt coefficients $\big\{ \alpha_i \big\}$. Then
    \[
      \big\| \ket{v} \big\|_{s(k)} = \sqrt{\sum_{i=1}^{k}\alpha_i^2}.
    \]
\end{thm}
\begin{proof}
  To see that $\big\| \ket{v} \big\|_{s(k)} \geq \sqrt{\sum_{i=1}^{k}\alpha_i^2}$, use the Schmidt Decomposition to write $\ket{v} = \sum_{i=1}^m \alpha_i \ket{u_i} \otimes \ket{v_i}$. Now let
  \[
    \ket{w} = \frac{\sum_{i=1}^k \alpha_i \ket{u_i} \otimes \ket{v_i}}{\sqrt{\sum_{i=1}^k \alpha_i^2}}.
  \]

  \noindent Observe that $SR(\ket{w}) \leq k$. Some algebra then reveals that
  \begin{align*}
    \braket{w}{v} & = \frac{1}{\sqrt{\sum_{i=1}^k \alpha_i^2}} \Big( \sum_{i=1}^{m}\alpha_i\bra{u_i} \otimes \bra{v_i} \Big) \Big(\sum_{i=1}^k \alpha_i \ket{u_i} \otimes \ket{v_i}\Big) \\
    & = \frac{1}{\sqrt{\sum_{i=1}^k \alpha_i^2}} \sum_{i=1}^{m}\sum_{j=1}^k \alpha_i \alpha_j \braket{u_i}{u_j} \otimes \braket{v_i}{v_j} \\
    & = \frac{1}{\sqrt{\sum_{i=1}^k \alpha_i^2}} \sum_{j=1}^k \alpha_j^2
     = \sqrt{\sum_{i=1}^k \alpha_i^2}.
  \end{align*}

  \noindent To see the opposite inequality, consider some fixed $\ket{w} \in \cl{H}_n \otimes \cl{H}_m$ with $SR(\ket{w}) \leq k$ and Schmidt Decomposition $\ket{w} = \sum_{i=1}^k \beta_i \ket{w_i} \otimes \ket{x_i}$. Then
  \begin{align*}
    \big| \braket{w}{v} \big| & = \left| \Big( \sum_{i=1}^{m}\alpha_i\bra{u_i} \otimes \bra{v_i} \Big) \Big(\sum_{i=1}^k \beta_i \ket{w_i} \otimes \ket{x_i}\Big) \right| \leq \sum_{i=1}^{m}\sum_{j=1}^k\alpha_i\beta_j\big|\braket{u_i}{w_j}\braket{v_i}{x_j}\big| = {\bf \alpha}^*D{\bf \beta}
  \end{align*}

\noindent where ${\bf \alpha}^{T} = (\alpha_1,\cdots,\alpha_m)$ and ${\bf \beta}^{T} = (\beta_1,\cdots,\beta_k,0,\cdots,0)$ are vectors of Schmidt coefficients, and $D$ is the matrix given by $D_{ij} = \big|\braket{u_i}{w_j}\braket{v_i}{x_j}\big|$ in which we have extended $\{\ket{u_i}\}$, $\{\ket{w_j}\}$ and $\{\ket{x_j}\}$ to orthonormal bases of their respective spaces. Observe that $D$ is doubly-sub-stochastic (i.e., each of its row and column sums is no greater than $1$) so the Hardy-Littlewood-Polya Theorem tells us that the vector ${\bf \gamma} := D{\bf \beta}$ satisfies
\begin{align*}
    \sum_{i=1}^j    \gamma_i \leq \sum_{i=1}^j \beta_i \quad \forall \, 1 \leq j \leq m.
\end{align*}

\noindent It follows from some simple linear algebra and the Cauchy-Schwarz inequality that
    \begin{align*}
    {\bf \alpha}^*D{\bf \beta} \leq {\bf \alpha}^*{\bf \beta} \leq \sqrt{\sum_{i=1}^k \alpha_i^2}\sqrt{\sum_{i=1}^k \beta_i^2} = \sqrt{\sum_{i=1}^k \alpha_i^2},
  \end{align*}
and the result follows.
\end{proof}

One useful way of looking at Theorem~\ref{thm:vectorNorm} is to notice that, because the Schmidt coefficients of $\ket{v}$ are the singular values of the operator $A_v$ to which $\ket{v}$ is associated in the proof of the Schmidt Decomposition Theorem, it follows that $\big\| \ket{v} \big\|_{s(k)}^2 = \big\| A_v^*A_v \big\|_k$.

The following is a straightforward consequence of Theorem~\ref{thm:vectorNorm} and is thus presented without proof.

\begin{cor}\label{cor:vecEquiv}
    Let $\ket{v} \in \cl{H}_n \otimes \cl{H}_m$ and suppose $h \leq k$. Then
    \begin{align*}
      \big\| \ket{v} \big\|_{s(h)} & \leq \big\| \ket{v} \big\|_{s(k)} \leq \sqrt{\frac{k}{h}} \big\| \ket{v} \big\|_{s(h)}.
    \end{align*}

    \noindent Furthermore, equality is achieved on the left if and only if $\big\| \ket{v} \big\|_{s(h)} = 1$ if and only if $SR(\ket{v}) \leq h$. Equality is achieved on the right if and only if the $k$ largest Schmidt coefficients of $\ket{v}$ are equal.
\end{cor}

\begin{remark}{\rm
Corollary~\ref{cor:vecEquiv} supports the interpretation of the $k$th vector norms discussed in Remark~\ref{Schmidtvectorrem1}, as it shows explicitly $\big\| \ket{v} \big\|_{s(k)} = 1$ (the largest possible value that norm can take on pure states) if and only if $SR(\ket{v}) \leq k$. On the other hand, consider the maximally-entangled state $\ket{e} := \frac{1}{\sqrt{n}}\sum_{i=1}^n\ket{e_i}\otimes\ket{e_i} \in \cl{H}_n \otimes \cl{H}_n$ -- Corollary~\ref{cor:vecEquiv} also implies that $\big\| \ket{e} \big\|_{s(k)} = \sqrt{\frac{k}{n}}$, which is the smallest the norm can ever be on pure states. We can make this interpretation of the vector norms more precise by using the trace distance and fidelity. It is not difficult to show via Equations~\eqref{eq:traceDistPure} and~\eqref{eq:fidelityPure} that
  \begin{align*}
    \sqrt{1 - \big\| \ket{v} \big\|_{s(k)}} & = \inf_{\ket{w}} \Big\{ \delta(\ketbra{v}{v},\ketbra{w}{w}) : SR(\ket{w}) \leq k \Big\} \text{ and } \\
    \big\| \ket{v} \big\|_{s(k)}^2 & = \sup_{\sigma} \Big\{ F(\ketbra{v}{v},\sigma) : SN(\sigma) \leq k \Big\}.
  \end{align*}}
\end{remark}

\section{Operator Norms}\label{sec:MatrixNorms}

In this section we define and investigate operator norms determined by the Schmidt decomposition. The vector norms discussed above are recovered in the special case of rank one operators, and will be used to derive an upper bound for the operator norms.
\begin{defn}\label{defn:operatorKNorm}
  Let $X \in \cl{L}(\cl{H}_n) \otimes \cl{L}(\cl{H}_m)$ and $1 \leq k \leq m$. Then we define the {\em $k$th operator norms} of $X$, denoted $\big\| X \big\|_{S(k)}$, by
  \begin{align*}
    \big\|X\big\|_{S(k)} & := \sup_{\ket{v},\ket{w}} \Big\{ \big|\bra{w}X\ket{v}\big| : SR(\ket{v}),SR(\ket{w}) \leq k \Big\}.
  \end{align*}
\end{defn}

Some minor observations are that $\big\| X \big\|_{S(m)} = \big\| X \big\|$ and $\big\| X \big\|_{S(k)} \leq \big\| X \big\|$ for all $k$. Further, in analogy with the vector norms, the operator norms form an increasing family of norms that lead up to the standard operator norm:
\begin{align*}
  \big\| X \big\|_{S(1)} \leq \big\| X \big\|_{S(2)} \leq \cdots \leq \big\| X \big\|_{S(m-1)} \leq \big\| X \big\|.
\end{align*}
Moreover, although $\big\| X^* \big\|_{S(k)} = \big\| X \big\|_{S(k)}$, it is {\em not} the case in general that $\big\| X^*X \big\|_{S(k)} = \big\| X \big\|_{S(k)}^2$. They also do not satisfy any natural submultiplicativity relationships.
  \begin{remark}{\rm
  Before continuing, let us comment briefly on this definition. One might notice that we could have just as well defined another generalization of the $k$th vector norms to the case of operators by treating $\cl{L}(\cl{H}_n) \otimes \cl{L}(\cl{H}_m)$ as a Hilbert space endowed with the Hilbert-Schmidt inner product. Then each operator $X \in \cl{L}(\cl{H}_n) \otimes \cl{L}(\cl{H}_m)$ can be thought of as a vector $x \in \cl{H}_{n^2} \otimes \cl{H}_{m^2}$ and we could define its  norm for $1 \leq k \leq m^2$ to be $\big\|x \big\|_{s(k)}$, the vector norm of the corresponding vector. However, one motivation for investigating the norm given by Definition~\ref{defn:operatorKNorm} instead is that the vector norms are in a sense trivial since they can be computed efficiently, as shown in Theorem~\ref{thm:vectorNorm}. Because the quantum separability problem is known to be NP-HARD \cite{G03}, and the problem of determining $k$-block positivity of an operator is believed to be very difficult, it seems unlikely that the vector norm of an operator can tell us very much about its block positivity or Schmidt number. On the other hand, we will see in Section~\ref{sec:SpectralInequalities} that the $k$th operator norm is a very powerful tool for detecting $k$-block positivity. Additionally, the operator norms of Definition~\ref{defn:operatorKNorm} build on the general principle that properties of pure states are easier to determine than properties of mixed states. We will see in Proposition~\ref{prop:rankOneNorm} that the operator norm for pure states reduces simply to the square of the vector norm of the corresponding pure vector state. Thus, the operator norms can efficiently be computed for pure states, but we will see that computing them for general mixed states is very difficult.
  }\end{remark}

The following result connects the  operator norms back to the vector norms in the pure state case, and shows that we can efficiently compute the operator norms when the operator under consideration has rank $1$. The proof easily follows from the relevant definitions and hence we leave it to the interested reader.
\begin{prop}\label{prop:rankOneNorm}
    Let $X = \ketbra{w}{v} \in \cl{L}(\cl{H}_n) \otimes \cl{L}(\cl{H}_m)$ be a rank-$1$ operator. Then
    \begin{align*}
      \big\| X \big\|_{S(k)} & = \big\| \ket{w} \big\|_{s(k)}\big\| \ket{v} \big\|_{s(k)}.
  \end{align*}
\end{prop}

We now present an important example to make use of Proposition~\ref{prop:rankOneNorm}.
\begin{exam}\label{exam:MatrixNormChoi}{\rm
  Recall the rank-$1$ projection operator $E = \frac{1}{n}\sum_{i,j=1}^n \ketbra{e_i}{e_j} \otimes \ketbra{e_i}{e_j} \in \cl{L}(\cl{H}_n) \otimes \cl{L}(\cl{H}_n)$. By Proposition~\ref{prop:rankOneNorm} we have that
  \begin{align*}
    \big\| E \big\|_{S(k)} & = \big\| \sum_{i=1}^n \frac{1}{\sqrt{n}}\ket{e_i} \otimes \ket{e_i} \big\|_{s(k)}^2 = \sum_{i=1}^{k}\left(\frac{1}{\sqrt{n}}\right)^2 = \frac{k}{n}.
  \end{align*}

  We will see that this simple example can be used to show that some inequalities that we derive in the next section are tight. It will also have applications to bound entanglement in Section~\ref{sec:WernerStates}.}
\end{exam}

The following proposition shows if $X$ is positive then it is enough to take the supremum only over $\ket{v}$ in the definition of the $k$th operator norms.
\begin{prop}\label{prop:MatMultDiffVectors}
  Let $X \in \cl{L}(\cl{H}_n) \otimes \cl{L}(\cl{H}_m)$ be positive semidefinite. Then
  \begin{align}\label{eq:MatMultDiffVectors}
    \big\| X \big\|_{S(k)} & = \sup_{\ket{v}} \big\{ \bra{v}X\ket{v} : SR(\ket{v}) \leq k \big\} \\ \label{eq:MatMultSN}
    & = \sup_{\rho} \big\{ \Tr(X \rho) : SN(\rho) \leq k \big\}.
  \end{align}
\end{prop}
\begin{proof}
  To show the first equality, write $X$ in its Spectral Decomposition as $X = \sum_i \lambda_i \ketbra{v_i}{v_i}$. Observe that the set of states $\ket{v}$ with $SR(\ket{v}) \leq k$ is compact, hence we can find a particular $\ket{v}$ with $SR(\ket{v}) \leq k$ such that $\sup_{\ket{v}} \{ \bra{v}X\ket{v} : SR(\ket{v}) \leq k \} = \sum_i \lambda_i |\braket{v_i}{v}|^2$. Then for any $\ket{w}$ with $SR(\ket{w}) \leq k$, we have that $\bra{w}X\ket{w} = \sum_i \lambda_i |\braket{v_i}{w}|^2 \leq \sum_i \lambda_i |\braket{v_i}{v}|^2$. Now define the $i^{th}$ component of two vectors $v^\prime$ and $w^\prime$ by $v_i^\prime := \sqrt{\lambda_i}|\braket{v_i}{v}|$ and $w_i^\prime := \sqrt{\lambda_i}|\braket{w}{v_i}|$. Applying the Cauchy-Schwarz inequality to $v^\prime$ and $w^\prime$ gives $|\bra{w}X\ket{v}| \leq \bra{v}X\ket{v}$. The other inequality is trivial. To see the second equality, simply write
  \[
    \sup_{\ket{v}} \big\{ \bra{v}X\ket{v} : SR(\ket{v}) \leq k \big\} = \sup_{\ket{v}} \big\{ \Tr(X\ketbra{v}{v}) : SR(\ket{v}) \leq k \big\},
  \]

  \noindent and note that the maximum on the right cannot become larger when taking the supremum over mixed states since a mixed state can be written as a convex combination of pure states.
\end{proof}

Equation~\eqref{eq:MatMultDiffVectors} captures a well-known property of the operator norm of positive operators in the $k = m$ case. We also note that Proposition~\ref{prop:MatMultDiffVectors} says that the $1$-norm, $\big\|\cdot\big\|_{S(1)}$, when acting on positive operators, coincides with the \emph{local spectral radius} $r^{loc}$ \cite{GPMSCZ09}. That is, if $X$ is positive then $\big\|X\big\|_{S(1)} = r^{loc}(X)$. Equation~\eqref{eq:MatMultSN} is perhaps a more natural way of looking at $\big\| X \big\|_{S(k)}$ from the quantum information perspective.

For a general mixed state $\rho$, one might want to think of $\big\| \rho \big\|_{S(k)}$ as measuring how close $\rho$ is to having Schmidt number of $k$ or less, but this interpretation is not quite right. Consider the following example, which shows that, in contrast to the $k$th vector norm case, it is not the case that $SN(\rho) \leq k$ implies $\big\|\rho\big\|_{S(k)} = \big\|\rho\big\|$.
\begin{exam}\label{ex:BadSNorm}
    {\rm Let $\rho \in \cl{L}(\cl{H}_2) \otimes \cl{L}(\cl{H}_2)$ have the following matrix representation in the standard basis $\{\ket{00},\ket{01},\ket{10},\ket{11}\}$:
    \[
        \rho = \frac{1}{5}\begin{bmatrix}2 & 1 & 1 & 1 \\ 1 & 1 & 1 & 1 \\ 1 & 1 & 1 & 1 \\ 1 & 1 & 1 & 1\end{bmatrix} = \frac{1}{5}\begin{bmatrix}1 & 0 \\ 0 & 0\end{bmatrix} \otimes \begin{bmatrix}1 & 0 \\ 0 & 0\end{bmatrix} + \frac{1}{5}\begin{bmatrix}1 & 1 \\ 1 & 1\end{bmatrix} \otimes \begin{bmatrix}1 & 1 \\ 1 & 1\end{bmatrix}.
    \]

    It is clear that $SN(\rho) = 1$. However, the eigenvector corresponding to the (distinct) maximal eigenvalue $0.8606$ is $\ket{v} := (0.6011, 0.4614, 0.4614, 0.4614)^T$. It is easily verified that $SR(\ket{v}) = 2$, so $\big\|\rho\big\|_{S(1)} < \big\|\rho\big\|$ (in fact, $\big\|\rho\big\|_{S(1)} \approx 0.8571$).}
\end{exam}

\begin{remark}
{\rm Nonetheless, it is the case that if the eigenspace corresponding to the maximal eigenvalue of $\rho$ contains a state $\ket{v}$ with $SR(\ket{v}) \leq k$ then $\big\|\rho\big\|_{S(k)} = \big\|\rho\big\|$. More importantly though, we can see via quantum fidelity that the correct interpretation of $\big\|\rho\big\|_{S(k)}$ is as a measure of how close $\rho$ is to a \emph{pure} state $\ket{v}$ with $SR(\ket{v}) \leq k$. More precisely, it is not difficult to show that
\begin{align*}
  \big\| \rho \big\|_{S(k)} & = \sup_{\ket{v}} \Big\{ F(\rho,\ketbra{v}{v}) : SR(\ket{v}) \leq k \Big\}.
\end{align*}}
\end{remark}

The following corollary shows that the $k$th operator norms are non-increasing under local quantum operations.
\begin{cor}\label{cor:LocalOps}
  Let $X \in \cl{L}(\cl{H}_n) \otimes \cl{L}(\cl{H}_m)$ be positive and let $\Phi :\cl{L}(\cl{H}_m) \rightarrow \cl{L}(\cl{H}_m)$ be a quantum channel (i.e. completely positive and trace-preserving). Then
  \[
    \big\| (id_n \otimes \Phi^\dagger)(X) \big\|_{S(k)} \leq \big\| X \big\|_{S(k)}.
  \]
\end{cor}
\begin{proof}
  By Proposition~\ref{prop:MatMultDiffVectors} we know that
  \begin{align*}
    \big\| (id_n \otimes \Phi^\dagger)(X) \big\|_{S(k)} & = \sup_{\rho} \Big\{ \Tr((id_n \otimes \Phi^\dagger)(X) \rho) : SN(\rho) \leq k \Big\} \\
    & = \sup_{\rho} \Big\{ \Tr(X (id_n \otimes \Phi)(\rho)) : SN(\rho) \leq k \Big\}.
  \end{align*}

  \noindent The result follows from the fact that Schmidt number is non-increasing under the action of local quantum channels \cite{TH00}, so $SN((id_n \otimes \Phi^\dagger)(\rho)) \leq k$.
\end{proof}

Finally, the last result of this section makes a crucial connection between the $k$th operator norms and $k$-block positivity of an operator.
\begin{cor}\label{cor:kPosInf1}
  Let $X \in (\cl{L}(\cl{H}_n) \otimes \cl{L}(\cl{H}_m))^{+}$ be positive and let $c \in \bb{R}$. Then $cI - X$ is $k$-block positive if and only if $c \geq \big\|X\big\|_{S(k)}$.
\end{cor}
\begin{proof}
    By Proposition~\ref{prop:kPos} we know that $cI - X$ is $k$-block positive if and only if
    \[
        \Tr((cI - X)\rho) = c - \Tr(X\rho) \geq 0 \quad \forall \, \rho \in \cl{L}(\cl{H}_n) \otimes \cl{L}(\cl{H}_m) \text{ with } SN(\rho) \leq k.
    \]

    \noindent Proposition~\ref{prop:MatMultDiffVectors} tells us that this is true precisely when $c \geq \big\|X\big\|_{S(k)}$.
\end{proof}

\begin{remark}
{\rm In particular, Corollary~\ref{cor:kPosInf1} shows that the problem of computing the operator norms is equivalent to the problem of determining $k$-block positivity of a Hermitian operator. Since the $k$-positivity problem is very difficult in general, computing these norms even just for positive operators must be a very difficult problem as well. Nevertheless, we shall see in the following sections that this connection leads to a new perspective for a number of different problems in quantum information.  }
\end{remark}

\subsection{Operator Norm Inequalities}\label{sec:MatrixNormProperties}

Since computing the $k$th operator norms in general is quite difficult, it will be useful to have explicitly calculable bounds for them. The following upper bound is thus of interest because it is explicitly computable in light of Theorem~\ref{thm:vectorNorm}.
  \begin{prop}\label{prop:matrixVectorNorms}
    Let $X \in \cl{L}(\cl{H}_n) \otimes \cl{L}(\cl{H}_m)$ be normal with eigenvalues $\{\lambda_i\}$ and corresponding eigenvectors $\{ \ket{v_i} \}$. Then
    \[
        \big\|X\big\|_{S(k)} \leq \sum_i |\lambda_i|\big\|\ket{v_i}\big\|^2_{s(k)}.
    \]
  \end{prop}
    \begin{proof}
        Let $\ket{v}$ and $\ket{w}$ have $SR(\ket{v}),SR(\ket{w}) \leq k$. Then we have
        \[
            \big| \bra{w}X\ket{v} \big| = \big| \sum_i \lambda_i \braket{w}{v_i}\braket{v_i}{v} \big| \leq \sum_i |\lambda_i| |\braket{w}{v_i} ||\braket{v_i}{v} | \leq \sum_i |\lambda_i | \big\| \ket{v_i} \big\|^2_{s(k)}.
        \]
    \end{proof}

Because $\cl{L}(\cl{H}_n) \otimes \cl{L}(\cl{H}_m)$ is finite-dimensional, we know that the $k$th operator norms must be equivalent. In order to quantify this fact, we will first need a simple lemma.
\begin{lemma}\label{lem:Schmidt01}
    Let $h \leq k$ and suppose $\ket{v} \in \cl{H}_n \otimes \cl{H}_m$ is a unit vector with $SR(\ket{v}) \leq k$. Then there exist nonnegative real constants $\{ d_j \}$ and (not necessarily distinct) unit vectors $\{ \ket{v_j} \} \subseteq \cl{H}_n \otimes \cl{H}_m$ for $1 \leq j \leq k$ such that $\sum_{j=1}^k d_j^2 = h$, $SR(\ket{v_j}) \leq h$, and
    \[
      h\ket{v} = \sum_{j=1}^k d_j\ket{v_j}.
  \]
\end{lemma}
\begin{proof}
  We can write $\ket{v}$ via the Schmidt Decomposition as $\ket{v} = \sum_{j=1}^k c_j \ket{a_j} \otimes \ket{b_j}$ with $\sum_{j=1}^k|c_j|^2 = 1$ and $\{ \ket{a_j} \}$, $\{ \ket{b_j} \}$ orthonormal sets. Thus
  \[
    h\ket{v} = \sum_{i=1}^h\sum_{j=1}^k c_j \ket{a_j} \otimes \ket{b_j}.
  \]

  Because $h \leq k$, we can rearrange the summations in such a way that we sum over $k$ sets of orthonormal vectors, with $h$ vectors in each set. We thus have $h\ket{v} = \sum_{j=1}^k d_j\ket{v_j}$ for some unit vectors $\ket{v_j}$ with $SR(\ket{v_j}) \leq h$ and constants $d_j$ satisfying $\sum_{j=1}^k d_j^2 = h$.
\end{proof}

\begin{thm}\label{thm:MatrixEquiv01}
    Let $X \in \cl{L}(\cl{H}_n) \otimes \cl{L}(\cl{H}_m)$ and suppose $h \leq k$. Then
    \begin{align*}
      \big\| X \big\|_{S(h)} & \leq \big\| X \big\|_{S(k)} \leq \frac{k}{h} \big\| X \big\|_{S(h)}.
    \end{align*}
\end{thm}
\begin{proof}
  The left inequality is trivial by the definition of the operator norms. To see the right inequality, suppose $\ket{v}$ and $\ket{w}$ have $SR(\ket{v}),SR(\ket{w}) \leq k$. Use Lemma~\ref{lem:Schmidt01} to write $h\ket{v} = \sum_{j=1}^k d_j\ket{v_j}$ and $h\ket{w} = \sum_{j=1}^k f_j\ket{w_j}$ so that
  \begin{align*}
    h^2\big| \bra{w}X\ket{v} \big| = \big| \sum_{i,j=1}^k{f_i d_j \bra{w_i} X \ket{v_j}} \big| \leq \Big(\sum_{i=1}^k{f_i}\Big)\Big(\sum_{i=1}^k{d_i}\Big) \big\| X \big\|_{S(h)} \leq kh \big\| X \big\|_{S(h)},
  \end{align*}

  \noindent where the rightmost inequality follows from two applications of the Cauchy-Schwarz inequality. The result follows by dividing through by $h^2$.
\end{proof}

To see that the inequalities of Theorem~\ref{thm:MatrixEquiv01} are tight, simply recall Example~\ref{exam:MatrixNormChoi}. Also observe that a straightforward consequence of this result is the inequality $\big\| X \big\|_{S(k)} \geq \frac{k}{m} \big\| X \big\|$ for all $k \leq m$. We now derive lower bounds that are much better in many situations.

\begin{prop}\label{prop:lowerBoundEig}
    Let $X = X^* \in \cl{L}(\cl{H}_n) \otimes \cl{L}(\cl{H}_m)$ have eigenvalues $\lambda_1 \leq \lambda_2 \leq \cdots \leq \lambda_{mn}$. Then for any $r \geq k$,
    \[
      \big\| X \big\|_{S(k)} \geq \frac{k\lambda_{mn - (n-r)(m-r)}}{r}.
    \]

    \noindent Furthermore, there exists an $X \in (\cl{L}(\cl{H}_n) \otimes \cl{L}(\cl{H}_m))^+$ such that $\big\| X \big\|_{S(k)} < \lambda_{nm - (n - k)(m - k) + 1}$.
\end{prop}
\begin{proof}
    Let $\cl{V}$ be the span of the eigenvectors $\ket{v_{nm - (n - r)(m - r)}}, \ket{v_{nm - (n - r)(m - r) + 1}}, \ldots, \ket{v_{nm}}$ corresponding to $\lambda_{nm - (n - r)(m - r)}, \lambda_{nm - (n - r)(m - r) + 1}, \ldots, \lambda_{mn}$. Then because ${\rm dim}(\cl{V}) = (n - r)(m - r) + 1$, by Theorem~\ref{thm:CMW08}, we know that there must exist a vector $\ket{v} \in \cl{V}$ with $SR(\ket{v}) \leq r$. It follows that
    \[
        \big\|X \big\|_{S(r)} \geq \big| \bra{v}X\ket{v} \big| \geq \sum_{i=nm - (n - r)(m - r)}^{mn} \lambda_i |\braket{v_i}{v}|^2 \geq \lambda_{nm - (n - r)(m - r)}.
    \]

    \noindent Using Theorem~\ref{thm:MatrixEquiv01} then shows that if $k \leq r$,
  \begin{align*}
    \big\| X \big\|_{S(k)} \geq \frac{k}{r}\big\| X \big\|_{S(r)} \geq \frac{k\lambda_{mn - (n-r)(m-r)}}{r}.
  \end{align*}

  To see the final claim, note that the dimension given by Theorem~\ref{thm:CMW08} is tight, so we can construct a positive operator $X$ with distinct eigenvalues such that the span of the eigenvectors corresponding to its $(n - k)(m - k)$ largest eigenvalues does not contain any states $\ket{w}$ with $SR(\ket{w}) \leq k$. It follows that $\bra{v}X\ket{v} < \lambda_{nm - (n - k)(m - k) + 1}$ for all $\ket{v}$ with $SR(\ket{v}) \leq k$.
\end{proof}

Now notice that if $P = P^* = P^2 \in (\cl{L}(\cl{H}_n) \otimes \cl{L}(\cl{H}_m))^+$ is an orthogonal projection, then by Theorem~\ref{thm:MatrixEquiv01} we have that $\frac{k}{m} \leq \big\| P \big\|_{S(k)} \leq 1$. The left inequality was seen to be tight by a rank-$1$ projection in Example~\ref{exam:MatrixNormChoi}, and it is not difficult to construct projection operators of any rank that have $\big\| P \big\|_{S(k)} = 1$. However, the following two results show that we can improve the lower bound if we take the rank of the projection, ${\rm rank}(P)$ into account.

\begin{thm}\label{thm:mainProjRes}
    Let $P = P^* = P^2 \in (\cl{L}(\cl{H}_n) \otimes \cl{L}(\cl{H}_m))^{+}$ be an orthogonal projection. Then
    \begin{align}\label{eq:projIneq1}
      \big\| P \big\|_{S(k)} & \geq \min\Big\{1,\frac{k}{\big\lceil \frac{1}{2}\big( n + m - \sqrt{(n-m)^2 + 4{\rm rank}(P) - 4} \big) \big\rceil}\Big\} \text{ and } \\ \label{eq:projIneq2}
      \big\|P\big\|_{S(k)} & \geq \frac{(k-1)mn + (m-k){\rm rank}(P)}{mn(m-1)}.
    \end{align}
\end{thm}
\begin{proof}
  To prove Inequality~\eqref{eq:projIneq1}, notice that Proposition~\ref{prop:lowerBoundEig} implies that $\big\| P \big\|_{S(r)} = 1$ whenever ${\rm rank}(P) \geq (n-r)(m-r) + 1$. Solving this inequality for $r$ gives
  \[
    r \geq \frac{1}{2}\Big( n + m - \sqrt{(n-m)^2 + 4{\rm rank}(P) - 4} \Big).
  \]

  Thus, choose $r = \Big\lceil \frac{1}{2}\Big( n + m - \sqrt{(n-m)^2 + 4{\rm rank}(P) - 4} \Big) \Big\rceil$ and $k \leq r$. Then using Proposition~\ref{prop:lowerBoundEig} again shows
  \[
      \big\| P \big\|_{S(k)} \geq \frac{k}{\big\lceil \frac{1}{2}\big( n + m - \sqrt{(n-m)^2 + 4{\rm rank}(P) - 4} \big) \big\rceil}.
    \]

    To show Inequality~\eqref{eq:projIneq2} holds, we first prove the result in the $k = 1$ case. Define $p := {\rm rank}(P)$. Choose orthonormal bases $\{ \ket{e_j} \}$ and $\{ \ket{f_l} \}$ of $\cl{L}(\cl{H}_n)$ and $\cl{L}(\cl{H}_m)$, respectively. Then choose $p$ orthonormal vectors $\ket{v_i}$ in the range of $P$ and observe that we can write them in the form
    \[
        \ket{v_i} = \sum_{j=1}^n\sum_{l=1}^m c_{ijl}\ket{e_j}\otimes\ket{f_l},
    \]

    \noindent where $\{c_{ijl}\} \in \bb{C}$ is a family of constants such that
    \begin{align}\label{eq:sumeq}
        \sum_{j=1}^n\sum_{l=1}^m |c_{ijl}|^2 = 1 \quad \forall \, i = 1, 2, \ldots, p.
    \end{align}

    \noindent It follows that there exists some fixed $j$ and $l$ such that
    \[
        \sum_{i=1}^p|c_{ijl}|^2 \geq \frac{p}{mn},
    \]

    \noindent since otherwise Equation~\eqref{eq:sumeq} would be violated. The $k = 1$ case follows by noting that, for this specific $j$ and $l$,
    \begin{align}\label{eq:s1LB}
        \big\|P\big\|_{S(1)} \geq (\bra{e_j} \otimes \bra{f_l})P(\ket{e_j} \otimes \ket{f_l}) = \sum_{i=1}^p\big|\bra{v_i}(\ket{e_j} \otimes \ket{f_l})\big|^2 = \sum_{i=1}^p|c_{ijl}|^2 \geq \frac{p}{mn}.
    \end{align}

    Now note that Theorem~\ref{thm:vectorNorm} says that for any $\ket{v}$ with $SR(\ket{v}) = 1$ and any $\ket{w} \in P\cl{H}$, $\big\| P \big\|_{S(1)} \geq \big| \braket{v}{w} \big|^2 \geq \alpha_1^2$, where $\alpha_1$ is the largest Schmidt coefficient of $\ket{w}$. On the other hand, it is clear that for any $\ket{w} \in P\cl{H}$, there exists a $\ket{v}$ with $SR(\ket{v}) = 1$ such that $\big| \braket{v}{w} \big|^2 = \alpha_1^2$. It follows that
    \begin{align}\label{eq:S1NormProj}
        \big\|P\big\|_{S(1)} = \sup_{\ket{w} \in P\cl{H}}\big\{ \alpha_1^2 : \alpha_1 \text{ is the largest Schmidt coefficient of } \ket{w} \big\}.
    \end{align}

    Now let $\ket{w} \in P\cl{H}$ have Schmidt coefficients $\{ \alpha_i \}$ such that $\alpha_1 = \sqrt{\big\| P \big\|_{S(1)}}$. Then using the facts that $\sum_{i=1}^m\alpha_i^2 = 1$ and $\alpha_i \geq \alpha_j$ for $i \leq j$, it follows that $\sum_{i=2}^m \alpha_i^2 = 1 - \big\| P \big\|_{S(1)}$ and so $\sum_{i=2}^k \alpha_i^2 \geq \frac{k-1}{m-1}(1 - \big\| P \big\|_{S(1)})$. Thus
    \[
        \big\|P\big\|_{S(k)} \geq \sum_{i=1}^k \alpha_i^2 = \big\|P\big\|_{S(1)} + \sum_{i=2}^k \alpha_i^2 \geq \big\|P\big\|_{S(1)} + \frac{(k-1)(1 - \big\|P\big\|_{S(1)})}{m-1}.
    \]

    \noindent The result follows by rearranging and using Inequality~\eqref{eq:s1LB}.
\end{proof}

Additionally, the same method as was used in the second half of the proof of Inequality~\eqref{eq:projIneq2} can be used to show the following improvement of the left inequality of Theorem~\ref{thm:MatrixEquiv01} in the case of projections.
\begin{cor}
    Let $P = P^* = P^2 \in (\cl{L}(\cl{H}_n) \otimes \cl{L}(\cl{H}_m))^+$ be an orthogonal projection and let $h \leq k$. Then
    \[
        \big\|P\big\|_{S(k)} \geq \left(1 - \frac{k - h}{m-1}\right)\big\|P\big\|_{S(h)} + \frac{k-h}{m-1}.
    \]
\end{cor}

\begin{remark}
{\rm Theorem~\ref{thm:mainProjRes} is particularly important because we will see that several important problems in quantum information theory could be answered if we were able to compute, or bound tightly, the $k$th operator norms of projections. Inequality~\eqref{eq:projIneq1} provides the best bound we have when ${\rm rank}(P)$ is small or large (e.g., ${\rm rank}(P) \leq m$ or ${\rm rank}(P) \geq (n-1)(m-1)$), but Inequality~\eqref{eq:projIneq2} is much tighter for moderate-rank projections (e.g., when ${\rm rank}(P) \approx \frac{mn}{2}$).

The two special cases of $k = 1$ and $k = m$ of Inequality~\eqref{eq:projIneq2} give lower bounds of $\frac{{\rm rank}(P)}{mn}$ and $1$, respectively -- the remaining lower bounds are just the linear interpolation of these two extremal cases. The bounds provided by Inequality~\eqref{eq:projIneq1} and Inequality~\eqref{eq:projIneq2} are used in the applications below. See \cite{JK09b} for a more detailed comparison of these inequalities.   }
\end{remark}

\section{Spectral Inequalities and Entanglement Witnesses}\label{sec:SpectralInequalities}

    In this section we derive a set of conditions for testing when a Hermitian operator is and is not $k$-block positive.
    The problem of determining $k$-block positivity is central to
    entanglement theory, as $k$-block positive operators can be
    used to detect the Schmidt number of mixed states via the
    theory of $k$-entanglement witnesses \cite{TH00,HHH96,KSu05}.

    Throughout this section, if $X = X^*$ then we will denote the positive eigenvalues of $X$ by $\{\lambda^{+}_i\}$ and the corresponding eigenvectors by $\{\ket{v^{+}_i}\}$. We will similarly denote the negative eigenvalues by $\{\lambda^{-}_i\}$ and the corresponding eigenvectors by $\{\ket{v^{-}_i}\}$, and the eigenvectors corresponding to the zero eigenvalues by $\{\ket{v^{0}_i}\}$. $X^{+} := \sum_i \lambda^{+}_i \ketbra{v^{+}_i}{v^{+}_i} \geq 0$ and $X^{-} := \sum_i \lambda^{-}_i \ketbra{v^{-}_i}{v^{-}_i} \leq 0$ are defined to be the positive and negative parts of $X$, respectively. Similarly, $P_{X}^{0} := \sum_i \ketbra{v^{0}_i}{v^{0}_i}$ and $P_{X}^{-} := \sum_i \ketbra{v^{-}_i}{v^{-}_i}$ denote the projections onto the nullspace and negative part of $X$, respectively.

\begin{thm}\label{thm:kposSpectral}
  Let $X = X^* \in \cl{L}(\cl{H}_n) \otimes \cl{L}(\cl{H}_m)$. Then
  \begin{enumerate}
    \item If $\big\| P_X^{-} \big\|_{S(k)} = 1$ then $X$ is not $k$-block positive.
    \item If $\big\| P_X^{0} + P_X^{-} \big\|_{S(k)} < 1$ and $\lambda_i^+ \geq \frac{\| X^{-} \|_{S(k)}}{1 - \| P_X^0 + P_X^{-} \|_{S(k)}}$ for all $i$, then $X$ is $k$-block positive.
        \item If $\big\| P_X^{-} \big\|_{S(k)} < 1$, all of the negative eigenvalues are equal, $X$ is nonsingular, and $\lambda_i^+ < \frac{\| X^{-} \|_{S(k)}}{1 - \| P_X^{-} \|_{S(k)}}$ for all $i$, then $X$ is not $k$-block positive.
    \end{enumerate}
\end{thm}
\begin{proof}
    To see statement (1), observe that there must be a vector $\ket{v} \in {\rm Range}(P_X^{-})$ such that $SR(\ket{v}) \leq k$. It follows that $\bra{v}X\ket{v} = \bra{v}X^{-}\ket{v} < 0$ and so $X$ is not $k$-block positive by Lemma~\ref{prop:kPos}.

    To see statement (2), let $\ket{v}$ be such that $SR(\ket{v}) \leq k$ and define $\mu := \frac{\| X^{-} \|_{S(k)}}{1 - \| P_X^0 + P_X^{-} \|_{S(k)}}$. Then, using the Spectral decomposition for $X^+$, the definition of the $k$th operator norm, and the hypotheses of (2), we have
    \begin{align*}
        \bra{v}X\ket{v} & = \bra{v}X^{+}\ket{v} - \big|\bra{v}X^{-}\ket{v}\big| \\ & \geq \sum_i \lambda_i^{+} |\braket{v}{v_i^{+}}|^2 - \big\|X^{-}\big\|_{S(k)} \\
         & \geq \mu \sum_i |\braket{v}{v_i^{+}}|^2 - \big\|X^{-}\big\|_{S(k)} \\ & \geq \mu (1 - \| P_X^0 + P_X^{-} \|_{S(k)}) - \big\|X^{-}\big\|_{S(k)} = 0,
    \end{align*}
and $X$ is $k$-block positive by Lemma~\ref{prop:kPos}.

    To see statement (3), observe that the set of unit vectors $\ket{v}$ with $SR(\ket{v}) \leq k$ is compact and so there exists a particular $\ket{v}$ with $SR(\ket{v}) \leq k$ such that $\big|\bra{v}X^{-}\ket{v}\big| = \big\|X^{-}\big\|_{S(k)}$. Define $\mu := \frac{\| X^{-} \|_{S(k)}}{1 - \| P_X^{-} \|_{S(k)}}$. Then similarly we have
    \begin{align*}
        \bra{v}X\ket{v} & = \bra{v}X^{+}\ket{v} - \big|\bra{v}X^{-}\ket{v}\big| \\ & = \sum_i \lambda_i^{+} |\braket{v}{v_i^{+}}|^2 - \big\|X^{-}\big\|_{S(k)} \\
         & < \mu \sum_i |\braket{v}{v_i^{+}}|^2 - \big\|X^{-}\big\|_{S(k)} \\ & = \mu (1 - \| P_X^{-} \|_{S(k)}) - \big\|X^{-}\big\|_{S(k)} = 0,
    \end{align*}
and again Lemma~\ref{prop:kPos} applies to show that $X$ is not $k$-block positive.
\end{proof}

\begin{remark}
{\rm On its face, Theorem~\ref{thm:kposSpectral} appears to be a
very technical result that may not be of much use due to the
difficulty of computing the $k$th operator norms. However, it is
not difficult to derive computable corollaries from it. In fact,
it implies a wide array of previously-known and new tests for
$k$-positivity and $k$-entanglement witnesses. These consequences
are presented below.


First, to see that Theorem~\ref{thm:kposSpectral} implies the
$k$-positivity results of Chru\'{s}ci\'{n}ski and Kossakowski
\cite{CKo09}, use Proposition~\ref{prop:matrixVectorNorms} and
simply note that their usage of the Ky Fan norm of Kraus operators
coincides with the $k$th norm of the corresponding
eigenvectors. It follows that this result also implies the
$k$-positivity test of Takesaki and Tomiyama \cite{TT83} and the
positivity test of Benatti, Floreanini, and Piani \cite{BFP04}, as
the tests of \cite{CKo09} do as well.  }
\end{remark}

Another corollary of this theorem is the following result of Kuah
and Sudarshan \cite{KSu05}.
\begin{cor}
  Suppose $\Phi : \cl{L}(\cl{H}_n) \rightarrow \cl{L}(\cl{H}_m)$ is a Hermicity-preserving linear map represented in its canonical Kraus representation $\Phi(\rho) = \sum_i\lambda_i^+ E_i\rho E_i^* + \sum_i\lambda_i^- F_i\rho F_i^*$, with the set of operators $\big\{ E_1, E_2, \ldots, F_1, F_2, \ldots \big\}$ forming an orthonormal basis in the Hilbert-Schmidt inner product. If ${\rm rank}(F_i) \leq k$ for some $i$, then $\Phi$ is not $k$-positive.
\end{cor}
\begin{proof}
    Simply recall that the Kraus operators $E_i$ and $F_i$ are exactly the operators to which the positive and negative eigenvectors of $X := (id_n \otimes \Phi)(E)$ are associated via the isomorphism used in the proof of the Schmidt Decomposition Theorem. Thus the rank of $F_i$ coincides with the Schmidt rank of the corresponding eigenvector $\ket{v_i}$.

    If $SR(\ket{v_i}) \leq k$ (i.e., ${\rm rank}(F_i) \leq k$) for some $i$ then $\big|\bra{v_i}P_X^{-}\ket{v_i}\big| = 1$ and so $\big\|P_X^{-}\big\|_{S(k)} = 1$. Condition (1) of Theorem~\ref{thm:kposSpectral} then gives the result.
\end{proof}

We have no references for the next two corollaries, though we expect they are well-known. The following corollary provides a characterization of the maximum number of negative eigenvalues that a $k$-block positive operator can have.
\begin{cor}\label{cor:maxKboxNegEval}
  Suppose $X = X^* \in \cl{L}(\cl{H}_n) \otimes \cl{L}(\cl{H}_m)$ is $k$-block positive. Then it has at most $(n - k)(m - k)$ negative eigenvalues.
\end{cor}
\begin{proof}
    Suppose $X$ has more than $(n - k)(m - k)$ negative eigenvalues. Then, by Theorem~\ref{thm:CMW08} it follows that there exists $\ket{v} \in {\rm Range}(P_X^{-})$ with $SR(\ket{v}) \leq k$. Hence we have $\big\| P_X^{-} \big\|_{S(k)} = 1$ and so condition (1) of Theorem~\ref{thm:kposSpectral} tells us that $X$ is not $k$-block positive.
\end{proof}

The following corollary shows just how negative the negative eigenvalues of a $k$-block positive operator can be.
\begin{cor}\label{cor:mostNegEval}
  Suppose $X = X^* \in \cl{L}(\cl{H}_n) \otimes \cl{L}(\cl{H}_m)$ is $k$-block positive. Denote the maximal and minimal eigenvalues of $X$ by $\lambda_{max}$ and $\lambda_{min}$, respectively. Then
  \[
    \frac{\lambda_{min}}{\lambda_{max}} \geq 1 - \frac{m}{k}.
    \]
\end{cor}
\begin{proof}
    If $\lambda_{min} \geq 0$ then the result is trivial. We thus assume that $\lambda_{min} < 0$. Suppose without loss of generality that $X$ has only one negative eigenvalue and is nonsingular (if this is not the case, we can add a suitable positive operator $Q$ to $X$ so that $X + Q$ is $k$-block positive, has a single negative eigenvalue equal to $\lambda_{min}$ and is nonsingular). If $X$ is $k$-block positive then condition (1) of Theorem~\ref{thm:kposSpectral} says that $\big\|P_X^{-}\big\|_{S(k)} < 1$. Condition (3) then says that
    \[
        \lambda_{max} \geq \frac{\| X^{-} \|_{S(k)}}{1 - \| P_X^{-} \|_{S(k)}} = -\lambda_{min}\frac{\| P_X^{-} \|_{S(k)}}{1 - \| P_X^{-} \|_{S(k)}}.
    \]

    Then
    \[
        \frac{\lambda_{min}}{\lambda_{max}} \geq \frac{\| P_X^{-} \|_{S(k)} - 1}{\| P_X^{-} \|_{S(k)}} = 1 - \frac{1}{\| P_X^{-} \|_{S(k)}} \geq 1 - \frac{m}{k}.
    \]
\end{proof}

By using Theorem~\ref{thm:mainProjRes} in the final step of the above
proof, we can derive the following bounds that in some sense
``interpolate'' between Corollary~\ref{cor:maxKboxNegEval} and
Corollary~\ref{cor:mostNegEval}, giving lower bounds on
$\lambda_{min}$ that depend on the number of negative eigenvalues
of $X$.
\begin{cor}
  Suppose $X = X^* \in \cl{L}(\cl{H}_n) \otimes \cl{L}(\cl{H}_m)$ is $k$-block positive with $r$ negative eigenvalues. Denote the maximal and minimal eigenvalues of $X$ by $\lambda_{max}$ and $\lambda_{min}$, respectively. Then
  \begin{align*}
    \frac{\lambda_{min}}{\lambda_{max}} & \geq 1 - \frac{\big\lceil \frac{1}{2}\big( n + m - \sqrt{(n-m)^2 + 4{r - 4}} \big) \big\rceil}{k} \ \ \text{ and} \\
    \frac{\lambda_{min}}{\lambda_{max}} & \geq 1 - \frac{mn(m-1)}{(k-1)mn + (m-k)r}.
    \end{align*}
\end{cor}

One final corollary shows that we now have a complete spectral characterization of the $k$-block positivity of Hermitian operators with exactly two distinct eigenvalues. The classification is trivial when both of the eigenvalues are negative or positive, but we believe that this is the first spectral classification for the case when they are of opposite signs.
\begin{cor}\label{cor:kposTwoEvals}
  Let $X = X^* \in \cl{L}(\cl{H}_n) \otimes \cl{L}(\cl{H}_m)$ have two distinct eigenvalues $\lambda_1 > \lambda_2$. Then $X$ is $k$-block positive if and only if
  \begin{align}\label{eq:kPosIneq}
    \| P_X^{-} \|_{S(k)} \leq \frac{\lambda_1}{\lambda_1 - \lambda_2}.
  \end{align}
\end{cor}
\begin{proof}
  If $\lambda_1$ and $\lambda_2$ have the same sign then the result is trivial. We thus assume that $\lambda_1 > 0$ and $\lambda_2 < 0$.

  If $X$ is $k$-block positive, then by condition (1) of Theorem~\ref{thm:kposSpectral} we know that $\| P_X^{-} \|_{S(k)} < 1$. Then condition (3) says that
  \[
    \lambda_1 \geq \frac{\| X^{-} \|_{S(k)}}{1 - \| P_X^{-} \|_{S(k)}} = -\lambda_2\frac{\| P_X^{-} \|_{S(k)}}{1 - \| P_X^{-} \|_{S(k)}}.
  \]

  \noindent The desired inequality follows easily. To see the other direction of the proof, suppose inequality~\eqref{eq:kPosIneq} is satisfied. Then because $\lambda_2 < 0$ it follows that $\| P_X^{-} \|_{S(k)} < 1$. $P_X^{0} = 0$ by hypothesis, so simple algebra shows that condition (2) of Theorem~\ref{thm:kposSpectral} is satisfied.
\end{proof}

\section{Bound Entanglement and Werner States}\label{sec:WernerStates}

One of the most pressing open questions in quantum information
theory is to find a classification of bound entangled states; that
is, states with zero distillable entanglement. A state $\rho$ is
\emph{distillable} if it can be transformed into the maximally
entangled state with only local operations and classical
communication \cite{DCL00}. If a state is separable then it has
positive partial transpose (PPT) \cite{P96}, and PPT states are
bound entangled \cite{H97,HHH98}. However, it is unknown whether
or not there exist states with NPPT that are bound entangled. In this section we use the $k$th operator norms to frame this fundamental question as a
concrete problem on a specific limit. We note that in \cite{CKS09}
it was shown that Chru\'{s}ci\'{n}ski and Kossakowski's
$k$-positivity tests could not be used to find entanglement
witnesses that detect non-positive partial transpose states, and
thus are not useful for trying to determine whether NPPT bound
entangled states exist.

It has been shown that a state $\rho$ is NPPT bound entangled if and only if $(\rho^\Gamma)^{\otimes k}$ is $2$-block positive for all $k \geq 1$ \cite{DSSTT00}. One especially important class of states in the study of bound entangled states is the family of Werner states \cite{W89}, which can be parametrized by a single real variable $\alpha \in [-1,1]$ via
\[
    \rho_\alpha := \frac{1}{n^2 - \alpha n}(I - \alpha nE^\Gamma) \in \cl{L}(\cl{H}_n) \otimes \cl{L}(\cl{H}_n).
\]
In particular, it is known that NPPT bound entangled states exist if and only if there is a Werner state that is NPPT bound entangled \cite{HH98}.

Because the partial transpose of Werner states have only two distinct eigenvalues (as noted in the following proof), Corollary~\ref{cor:kposTwoEvals} applies to this situation and the Schmidt operator norms are a natural tool for tackling this problem. The following result is a starting point.
\begin{prop}\label{prop:werner01}
    Let $\rho_\alpha \in \cl{L}(\cl{H}_n) \otimes \cl{L}(\cl{H}_n)$ be a Werner state. Then $\rho_\alpha^\Gamma$ is $k$-block positive if and only if $\alpha \leq \frac{1}{k}$.
\end{prop}
\begin{proof}
    Simply note that $(n^2 - \alpha n)\rho_\alpha^\Gamma = I - \alpha nE$ has only two distinct eigenvalues: $1$ and $1 - \alpha n$. Corollary~\ref{cor:kposTwoEvals} then implies that $\rho_\alpha^\Gamma$ is $k$-block positive if and only if $\big\|E\big\|_{S(k)} \leq \frac{1}{\alpha n}$. We saw in Example~\ref{exam:MatrixNormChoi} that $\big\|E\big\|_{S(k)} = \frac{k}{n}$, so the result follows.
\end{proof}

\begin{remark}
{\rm The special case $k = n$ of the above proposition is very
well-known and states that $\rho_\alpha$ is PPT if and only if
$\alpha \leq \frac{1}{n}$. Moreover, Proposition~\ref{prop:werner01}
shows that Werner states can not be bound entangled for $\alpha >
\frac{1}{2}$, which is also well-known. It has been conjectured
\cite{DSSTT00,DCL00} that Werner states are bound entangled for
all $\alpha \leq \frac{1}{2}$; this is exactly the set of values
for which $\rho_\alpha^\Gamma$ is $2$-positive.

Although we now have determined $k$-block positivity of $\rho_\alpha^\Gamma$, determining $k$-block positivity of $(\rho_\alpha^\Gamma)^{\otimes r}$ for $r > 1$ is not so simple in general because the projection onto the negative eigenspaces is no longer rank-$1$, so we cannot exactly compute its $k$th norm. Additionally, $(\rho_\alpha^\Gamma)^{\otimes r}$ has more than two distinct eigenvalues in general so we can no longer use Corollary~\ref{cor:kposTwoEvals}. To simplify the problem somewhat, consider the $\alpha = \frac{2}{n}$ case. Then the operator $X := (n^2 - 2)\rho_{2/n} = I - 2E^\Gamma$ has eigenvalues $1$ and $-1$, so $(X^\Gamma)^{\otimes r}$ has only two distinct eigenvalues ($1$ and $-1$) regardless of $r$. Corollary~\ref{cor:kposTwoEvals} then says that $\rho_{2/n}$ is bound entangled if and only if $\big\| P_r^{-} \big\|_{S(2)} \leq \frac{1}{2}$ for all $r \geq 1$, where $P_r^{-}$ is the projection onto the $-1$ eigenspace of $(\rho_{2/n}^\Gamma)^{\otimes r}$. This mirrors the approach attempted in \cite{PPHH07} to find a bound entangled NPPT Werner state, though that paper considers the $n = 4$ case exclusively. Note in particular that our tests of $k$-positivity derived in the previous section are strong enough to determine bound entanglement in some cases, assuming we can compute or find strong bounds on these norms in this situation.

We will finish this section by showing that, in the limit as $r$ tends to infinity, it is not possible to do any better than $\big\| P_r^{-} \big\|_{S(2)} \leq \frac{1}{2}$. More precisely, it is the case that
\[
    \lim_{r \to \infty} \big\| P_r^{-} \big\|_{S(2)} \geq \frac{1}{2}.
\]

To prove this claim, observe that ${\rm rank}(P_1^{-}) = 1$ and $P_{r}^{-} = P_{1}^{-} \otimes P_{r-1}^{+} + P_{1}^{+} \otimes P_{r-1}^{-}$ for all $r \geq 2$, where $P_{r}^{+}$ is the projection onto the $+1$ eigenspace of $(\rho_{2/n}^\Gamma)^{\otimes r}$. It follows that ${\rm rank}(P_r^{-}) = {\rm rank}(P_{r-1}^{+}) + (n^2 - 1){\rm rank}(P_{r-1}^{-})$ for all $r \geq 2$. Standard techniques for solving recurrence relations then show that ${\rm rank}(P_r^-) = \frac{1}{2}(n^{2r} - (n^2 - 2)^r)$ for all $r \geq 1$. Plugging this into the lower bound of Inequality~\eqref{eq:projIneq2} reveals that
\[
    \big\|P_r^-\big\|_{S(2)} \geq \frac{n^{2r} + \frac{1}{2}(n^r-2)(n^{2r} - (n^2 - 2)^r)}{n^{2r}(n^r-1)} = \frac{n^r-2}{2(n^r-1)} - \frac{(n^r-2)(n^2 - 2)^r - 2n^{2r}}{2n^{2r}(n^r-1)}.
\]
It is not difficult to verify that the lower bound on the right is always, for $n \geq 4$, strictly less than $\frac{1}{2}$. Furthermore, as $r \rightarrow \infty$, the rightmost fraction tends to zero and the left fraction tends to $\frac{1}{2}$. This shows that, asymptotically, $\frac{1}{2}$ is the smallest that we could ever hope $\big\|P_r^-\big\|_{S(2)}$ to be. Thus we have proved the following.  }
\end{remark}

\begin{cor}
The Werner state $\rho_{2/n}$ is bound entangled if and only if
\[
    \lim_{r \to \infty} \big\| P_r^{-} \big\|_{S(2)} = \frac{1}{2}.
\]
\end{cor}

\section{Outlook}\label{sec:outlook}

We have seen that the family of  norms studied here play an important role in quantum information theory and have actually been used implicitly several times over the past decade. They are powerful tools for determining $k$-positivity of Hermitian operators, and hence entanglement witnesses, especially for the partial transpose of Werner states and other operators with only two distinct eigenvalues. Further exploration of the relationship between these norms, $k$-positivity tests, and Werner states in search of NPPT bound entangled states is warranted.

Many of the applications of these norms involve only the value of the norm on orthogonal projections. While we derived several ways to bound these norms, we do not know if our best lower bound involving $n$, $k$, and the rank of the projection is tight. Thus, a tight lower bound would be of significant interest, as would a characterization of the projections that attain the lower bound. Further analysis of the associated computational issues from a semidefinite programming perspective is included in \cite{JK09b}.

\vspace{0.1in}

\noindent{\bf Acknowledgements.} We are grateful to the Fields Institute for hosting the Thematic Program on Mathematics in Quantum Information, which led to this project. We thank Fernando Brandao, Vern Paulsen, and Mary Beth Ruskai for helpful comments and conversations. Thanks are also extended to Stanislaw Szarek, Elisabeth Werner, and Karol {\.Z}yczkowski for drawing our attention to a minor error in an earlier version of the paper. N.J. was supported by an NSERC Canada Graduate Scholarship and the University of Guelph Brock Scholarship. D.W.K. was supported by NSERC Discovery Grant 400160, NSERC Discovery Accelerator Supplement 400233, and Ontario Early Researcher Award 048142.


\end{document}